\documentclass{llncs}
\usepackage[english]{babel}
\usepackage[left=2cm,right=2cm,top=3cm,bottom=3cm]{geometry}
\usepackage{amsmath,amssymb,mathrsfs, amsfonts}
\usepackage{color}
\usepackage{ifthen}
\usepackage[pdftex]{graphicx}
\usepackage{stmaryrd}
\usepackage[
  pdftex,
  bookmarks=true,
  bookmarksnumbered=true,
  hypertexnames=false,
  breaklinks=true,
  linkbordercolor={0 0 1}
  ]{hyperref}

\newcommand{\importfigures}{\boolean{true}}
\newcommand{\generatefigures}{\boolean{false}}

\ifthenelse{\importfigures}{}{%
\usepackage{tikz}%
\usepgflibrary{arrows}%
\usepackage{environ}%
}

\ifthenelse{\generatefigures}{%
\usetikzlibrary{external}%
\tikzexternalize\tikzset{external/force remake}%
}{}

\newcommand{\R}{\mathbb{R}}
\newcommand{\Q}{\mathbb{Q}}
\newcommand{\N}{\mathbb{N}}

\newcommand{\proj}[2]{\pi_{#1}(#2)}
\newcommand{\defref}[1]{definition~\ref{#1}}
\newcommand{\lemref}[1]{lemma~\ref{#1}}
\newcommand{\thref}[1]{theorem~\ref{#1}}
\newcommand{\corref}[1]{corollary~\ref{#1}}
\newcommand{\figref}[1]{figure~\ref{#1}}
\newcommand{\figrefmult}[1]{figures~{#1}}
\newcommand{\secref}[1]{section~\ref{#1}}

\newcommand{\sgn}[1]{\operatorname{sgn}(#1)}

\newcommand{\intp}{\operatorname{int}}
\newcommand{\fracp}{\operatorname{frac}}
\newcommand{\fiter}[2]{{#1}^{[#2]}}
\newcommand{\GSpace}[2]{\operatorname{GSPACE}\left(#1,#2\right)}
\newcommand{\lambdafun}[2]{(#1)\mapsto{#2}}

\newlength{\bracketwidth}

\newcommand{\inorm}[2]{{\left\lVert{#1}\right\rVert_{#2}}}
\newcommand{\twonorm}[1]{\inorm{#1}{2}}
\newcommand{\infnorm}[1]{\inorm{#1}{}}
\newcommand{\intinterv}[2]{\{#1,\ldots,#2\}}

\newcommand{\fnxi}[3]{tanh((#1)*(#2)*(#3))}
\newcommand{\fnsigmaone}[3]{(1+\fnxi{#1-1}{#2}{#3})/2}
\newcommand{\fnsigmanomult}[4]{%
    \ifthenelse{\numexpr#1\relax=1}%
    {\fnsigmaone{#2}{#3}{#4}}%
    {\fnsigmaone{#2-#1+1}{#3}{#4}+\fnsigmanomult{#1-1}{#2}{#3}{#4}}%
}
\newcommand{\fnsigmathree}[3]{%
    \fnsigmaone{#1}{#2+log(3)}{#3}+\fnsigmaone{#1-1}{#2+log(3)}{#3}+%
    \fnsigmaone{#1-2}{#2+log(3)}{#3}}
\newcommand{\fnsigmafour}[3]{%
    \fnsigmaone{#1}{#2+log(4)}{#3}+\fnsigmaone{#1-1}{#2+log(4)}{#3}+%
    \fnsigmaone{#1-2}{#2+log(4)}{#3}+\fnsigmaone{#1-3}{#2+log(4)}{#3}}
\newcommand{\fntheta}[2]{exp(-(#2)*(1-sin(2*pi*(#1)))**2)}

\allowdisplaybreaks

\begin{document}

\author{Olivier Bournez\inst{1} \and Daniel S. Gra\c{c}a\inst{2,3}
  \and Amaury Pouly\inst{1,2}}
\institute{
Ecole Polytechnique, LIX, 91128 Palaiseau Cedex, France. \\
\email{Olivier.Bournez@lix.polytechnique.fr}
\and  CEDMES/FCT, Universidade do Algarve, C. Gambelas, 8005-139 Faro, Portugal.
\email{dgraca@ualg.pt}
\and
SQIG /Instituto de Telecomunica\c{c}\~{o}es, Lisbon, Portugal.
}

\title{Turing machines can be efficiently simulated by the General Purpose Analog Computer}

\maketitle

\begin{abstract}
The Church-Turing thesis states that any sufficiently powerful computational model which captures the notion of algorithm is computationally equivalent to the Turing machine. This equivalence usually holds both at a computability level and at a computational complexity level modulo polynomial reductions. However, the situation is less clear in what concerns models of computation using real numbers, and no analog of the Church-Turing thesis exists for this case. Recently it was shown that some models of computation with real numbers were equivalent from a computability perspective. In particular it was shown that Shannon's General Purpose Analog Computer (GPAC) is equivalent to Computable Analysis. However, little is known about what happens at a computational complexity level. In this paper we shed some light on the connections between this two models, from a computational complexity level, by showing that, modulo polynomial reductions, computations of
Turing machines can be simulated by GPACs, without the need of using
more (space) resources than those used in the original Turing
computation, as long as we are talking about bounded computations. In other words, computations done by the GPAC are as space-efficient as computations done in the context of Computable Analysis.
\end{abstract}

\section{Introduction}
The Church-Turing thesis is a cornerstone statement in theoretical
computer science, stating that any (discrete time, digital)
sufficiently powerful computational
model which captures the notion of algorithm is computationally
equivalent to the Turing machine (see e.g.~\cite{Odi89},
\cite{Sip05}). It also relates various aspects of models in a very
surprising and strong way.

The Church-Turing thesis, although not formally a theorem, follows from many equivalence results for discrete models and is  considered to be valid by the scientific community \cite{Odi89}.
When considering non-discrete time or non-digital models, the
situation is far from being so clear. In particular, when considering
models working over real numbers, several models are clearly not
equivalent \cite{Bra00}.

However, a question of interest is whether physically \emph{realistic}
models of computation over the real numbers are equivalent, or can be
related. Some of the results of non-equivalence involve models, like
the BSS model \cite{BSS89}, \cite{BCSS98}, which are claimed not to be
physically realistic \cite{Bra00} (although they certainly are
interesting from an algebraic perspective),  or models that depend
critically of computations which use exact numbers to obtain
super-Turing power, e.g.~\cite{AM98}, \cite{BBKT01}.

Realistic models of computation over the reals clearly include the
\emph{General Purpose Analog Computer (GPAC)}  \cite{Sha41}, an analog
continuous-time model of computation and \emph{Computable Analysis} (see e.g. \cite{Wei00}).  The GPAC is a mathematical
model introduced by Claude Shannon of an earlier analog
computer, the Differential Analyzer. The first general-purpose
Differential Analyzer is generally attributed to Vannevar Bush
\cite{Bus31}. Differential Analyzers have been used intensively up to
the 1950's as computational machines to solve various problems from
ballistic to aircraft design, before the era of the digital computer
\cite{Nyc96}.

Computable analysis, based on Turing machines, can be considered as
today's most used model for talking about computability and complexity
over reals. In this approach, real numbers are encoded as sequences of
discrete quantities and a discrete model is used to compute over these
sequences. More details can be found in the books \cite{PR89},
\cite{Ko91}, \cite{Wei00}. As this model is based on classical
(digital and discrete time) models like Turing machines, which are considered to be realistic models of today's computers, one
can consider that Computable Analysis is a realistic model (or, more correctly, a theory)
of computation.

Understanding whether there could exist something similar to a
Church-Turing thesis models of computation involving real numbers, or whether
analog models of computation could be more powerful than today's
classical models of computation motivated us to try to relate GPAC computable
functions to functions computable in the sense of computable analysis.

The paper \cite{BCGH07} was a first step towards the objective of
obtaining a version of the Church-Turing thesis for physically
feasible models over the real numbers. This paper proves that, from a
computability perspective, Computable Analysis and the GPAC are
equivalent: GPAC computable functions are computable and, conversely,
functions computable by Turing machines or in the computable analysis
sense can be computed by GPACs.
However this is about \emph{computability}, and not
\emph{computational complexity}. This proves that one cannot solve more problems
using the GPAC than those we can solve using discrete-based approaches such as Computable Analysis. But this leaves open the question
whether one could solve some problems \emph{faster} using analog
models of computations (see e.g. what happens for quantum models of
computations\dots).   In other words, the question of whether the
above models are equivalent at a computational complexity level
remained open. Part of the difficulty stems from finding an
appropriate notion of complexity (see e.g.~\cite{SBF99}, \cite{BSF02})
for analog models of computations.

In the present paper we study both
the GPAC and Computable Analysis at a complexity level. In particular,
we introduce measures for space complexity and show that, using these
measures, both models are equivalent, even at a computational
complexity level, as long as we consider time-bounded simulations. Since we already have shown in our previous paper
\cite{BGP12} that Turing machines can simulate efficiently GPACs, this paper is a big step
towards showing the converse direction: GPACs can simulate Turing
machines in an efficient manner.

More concretely we show that computations of Turing machines can be simulated in polynomial space by GPACs as long as we use bounded (but arbitrary) time. We firmly believe that this construction can be used as a building brick to show the more general result that the computations of Turing machines can be simulated in polynomial space by GPACs, removing the hypothesis of arbitrary but fixed time. This latter construction would probably be much more involved, and we intend to focus on it in the near future since this result would show that computations done by the GPAC and in the context of Computable Analysis are equivalent modulo polynomial space reductions.

We believe that these results open the way for some sort of more general Church-Turing thesis, which applies not only to discrete-based models of computation but also to physically realistic models of computation, and which holds both at a computability and computational complexity (modulo polynomial reductions) level.

Incidently, these kind of results can also be the first step towards a well-founded complexity
theory for analog models of computations and for continuous dynamical
systems.

Notice that it has been observed in several papers that, since
continuous time systems might undergo arbitrary space and time
contractions, Turing machines, as well as even accelerating Turing
machines\footnote{Similar possibilities of simulating accelerating
  Turing machines through quantum mechanics are discussed in
  \cite{CP01}.} \cite{Davies01}, \cite{Cop98}, \cite{Cop02} or
even oracle Turing machines, can actually be simulated in an arbitrary
short time by ordinary differential equations in an arbitrary short
time or space. This is sometimes also called \textit{Zeno's
phenomenon}: an infinite number of discrete transitions may happen in
a finite time: see e.g. \cite{CIEChapter2007}.
Such constructions or facts have been deep obstacles
to various attempts to build a well founded complexity theory for
analog models of computations: see \cite{CIEChapter2007} for
discussions.
One way to interpret our results is then the following: all these time and
space phenomena, or Zeno's phenomena do not hold (or, at least, they do not hold in a
problematic manner) for ordinary differential equations
corresponding to GPACs, that is to say for \emph{realistic}
models, for carefully chosen measures of complexity. Moreover, these measures of complexity relate naturally to standard computational complexity measures involving discrete models of computation

\section{Preliminaries}

\subsection{Notation}

Throughout the paper we will use the following notation:
\[\infnorm{(x_1,\ldots,x_n)}=\max_{1\leqslant i\leqslant n}|x_i|\qquad
\twonorm{(x_1,\ldots,x_n)}=\sqrt{|x_1|^2+\cdots+|x_n|^2}\]
\[\proj{i}{x_1,\ldots,x_k}=x_i\qquad
\intp(x)=\lfloor x\rfloor\qquad
\fracp(x)=x-\lfloor x\rfloor\]
\[\intp_n(x)=\min(n,\intp(x))\qquad
\fracp_n(x)=x-\intp_n(x)\]
\[\fiter{f}{n}=
\begin{cases}
\operatorname{id} & \text{if }n=0\\
\fiter{f}{n-1} & \text{otherwise}
\end{cases}\]
\[\sgn{x}=\begin{cases}-1&\text{if }x<0\\0&\text{if }x=0\\1&\text{if
  }x>0\end{cases}\qquad 
\R^*=\R\setminus\{0\}\]

\subsection{Computational complexity measures for the GPAC}\label{section:GPAC}

It is known \cite{GC03} that a function is generable by a GPAC iff
it is a component of the solution a polynomial initial-value
problem. In other words, a function $f:I\rightarrow\R$ is
GPAC-generable  iff it belongs to following class. 

\begin{definition}\label{def:gpac_class}
Let $I\subseteq\R$ be an open interval and $f:I\rightarrow\R$. We say that $f\in\operatorname{GPAC}(I)$ if there exists $d\in\N$, a vector of polynomials $p$, $t_0\in I$ and $y_0\in\R^d$ such that for all $t\in I$ one has $f(t)=y_1(t)$, where $y:I\rightarrow\R$ is the unique solution over $I$ of
\begin{equation}
\left\{\begin{array}{@{}c@{}l}\dot{y}&=p(y)\\y(t_0)&=y_0\end{array}\right.
\end{equation}\label{eq:ode}
\end{definition}

Next we introduce a subclass of GPAC generable functions which
allow us to talk about space  complexity. The idea is that a function $f$ generated by a GPAC belongs to the class $\GSpace{I}{g}$ if $f$ can be generated by a GPAC in $I$ and does not grow faster that $g$. Since the value of $f$ in physical implementations of the GPAC correspond to some physical quantity (e.g.~electric tension), limiting the growth of $f$ corresponds to effectively limiting the size of resources needed to compute $f$ by a GPAC.

\begin{definition}\label{def:gspace_class}
Let $I\subseteq\R$ be an open interval and $f,g:I\rightarrow\R$ be functions. The function $f$
belongs to the class $\GSpace{I}{g}$ if there exist $d\in\N$, a vector of polynomials $p$, $t_0\in I$ and $y_0\in\R^d$ such that for all $t\in I$ one has $f(t)=y_1(t)$ and $\infnorm{y(t)}\leqslant g(t)$, where $y:I\rightarrow\R$ is the unique solution over $I$ of \eqref{eq:ode}. More generally, a function $f:I\rightarrow\R^d$ belongs to $f\in\GSpace{I}{g}$ if all its components are also in the same class.
\end{definition}

We can generalize the complexity class GSPACE to multidimensional open sets $I$ defined over $\R^d$. The idea is to reduce it to the one-dimensional case defined above through the introduction of a subset $J\subseteq\R$ and of a map $g:J\to I$.
\begin{definition}\label{def:gspace_ext_class}
Let $I\subseteq\R^d$ be an open set and $f,s_f:I\rightarrow\R$ be
functions. Then $f\in\GSpace{I}{s_f}$ if for any open interval
$J\subseteq\R$ and any function ($g:J\rightarrow\R^d \in\GSpace{J}{s_g}\text{ such that }g(J)\subseteq I$, one has $f\circ g\in\GSpace{J}{\max(s_g,s_f\circ s_g)}.$
\end{definition}

The following closure results can be proved (proofs are omitted for reasons of space).
\begin{lemma}\label{lem:gspace_ext_class_stable}
Let $I,J\subseteq\R^d$ be open sets, and $(f:I\rightarrow\R^n)$ and ($g:J\rightarrow\R^m$) be functions which belong to $\GSpace{I}{s_f}$ and $\GSpace{J}{s_g}$, respectively. Then:
\begin{itemize}
\item $f+g, f-g\in\GSpace{I\cap J}{s_f+s_g}$ if $n=m$.
\item $fg\in\GSpace{I\cap J}{\max(s_f,s_g,s_fs_g)}$ if $n=m$.
\item $f\circ g\in\GSpace{J}{\max(s_g,s_f\circ s_g)}$ if $m=d$ and $g(J)\subseteq I$.
\end{itemize}
\end{lemma}

\subsection{Main result}

Our main result states that any Turing machine can be simulated by a
GPAC using a space bounded by a polynomial, where $T$ and $S$ are respectively the
time and the space used by the Turing machine.


If one prefers, (formal statement in Theorem
\ref{th:tm_simual_diff_eq}): 

\begin{theorem}
Let $\mathcal{M}$ be a Turing Machine. Then there is a GPAC-generable function $f_M$ and a polynomial $p$ with the following properties:
\begin{enumerate}
\item Let $S,T$ be arbitrary positive integers. Then $f_\mathcal{M}(S,T,[e],n)$ gives  the configuration of $\mathcal{M}$ on
input $e$ at step $n$, as long as $n\leq T$ and $\mathcal{M}$ uses space bounded by $S$.
\item $f_\mathcal{M}(S,T,[e],t)$ is bounded by $p(T+S)$ as long as $0\leq t\leq n$.
\end{enumerate}
\end{theorem}

The first condition of the theorem states that the GPAC simulates TMs on bounded space and time, while the second condition states that amount of resources used by the GPAC computation is polynomial on the amount of resources used by original Turing computation. 

\section{The construction}

\subsection{Helper functions}

Our simulation will be performed on a real domain and may be subject to (small) errors. Thus, to simulate a Turing machine over a large number of steps, we need tools which allow us to keep errors under control. In this section we present functions which are specially designed to fulfill this objective. We call these functions \emph{helper functions}. Notice that since functions generated by GPACs are analytic, all helper functions are required to be analytic. As a building block for creating more complex functions, it will be useful to obtain analytic approximations of the functions $\intp(x)$ and $\fracp(x)$. Notice that we are only concerned about nonnegative numbers so there is no need to discuss the definition of these functions on negative numbers. A graphical representation of the various helper functions we will introduce in this section can be found on \figrefmult{\ref{fig:xi},\ref{fig:sigma} and \ref{fig:theta}}. Proofs within this section are ommited for reasons of space.

\begin{definition}\label{def:xi}
For any $x,y,\lambda\in\R$ define
$\xi(x,y,\lambda)=\tanh(xy\lambda)$. 
\end{definition}

\begin{lemma}\label{lem:xi}
For any $x\in\R$ and $\lambda>0,y\geqslant1$,
\[|\sgn{x}-\xi(x,y,\lambda)|<1\]
Furthermore if $|x|\geqslant\lambda^{-1}$ then
\[|\sgn{x}-\xi(x,y,\lambda)|<e^{-y}\]
and $\xi\in\GSpace{\R^{3}}{1}$.
\end{lemma}

\begin{figure}
\begin{center}
\ifthenelse{\importfigures}{%
\includegraphics{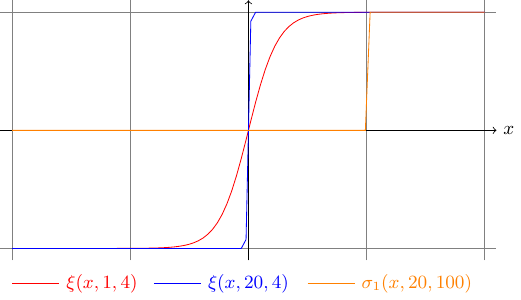}}{%
\begin{tikzpicture}[domain=-2:2,samples=100,scale=2]
\draw[very thin,color=gray] (-2.1,-1.1) grid (2.1,1.1);
\draw[->] (0,-1.1) -- (0,1.1);
\draw[->] (-2.1,0) -- (2.1,0) node[right] {$x$};
\draw[color=red] plot[id=fn_xi_1] function{\fnxi{x}{1}{4}};
\draw[color=blue] plot[id=fn_xi_2] function{\fnxi{x}{20}{4}};
\draw[color=orange] plot[id=fn_sigma1_1] function{\fnsigmaone{x}{20}{100}};
\draw[color=red] (-2,-1.3) -- (-1.6,-1.3) node[right] {$\xi(x,1,4)$};
\draw[color=blue] (-0.8,-1.3) -- (-0.4,-1.3) node[right] {$\xi(x,20,4)$};
\draw[color=orange] (0.5,-1.3) -- (0.9,-1.3) node[right] {$\sigma_1(x,20,100)$};
\end{tikzpicture}}
\end{center}
\caption{Graphical representation of $\xi$ and $\sigma_1$}
\label{fig:xi}
\end{figure}
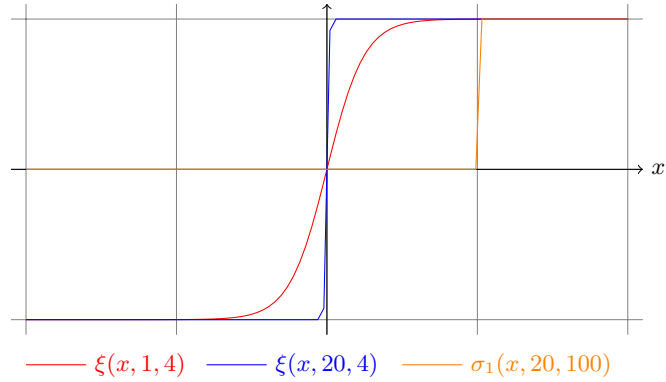

\begin{definition}
For any $x,y,\lambda\in\R$, define
\[\sigma_1(x,y,\lambda)=\frac{1+\xi(x-1,y,\lambda)}{2}\]
\end{definition}

\begin{corollary}\label{cor:sigma1}
For any $x\in\R$ and $y>0,\lambda>2$,
\[|\intp_1(x)-\sigma_1(x,y,\lambda)|\leqslant1/2\]
Furthermore if $|1-x|\geqslant\lambda^{-1}$ then
\[|\intp_1(x)-\sigma_1(x,y,\lambda)|<e^{-y}\]
and $\sigma_1\in\GSpace{\R^{3}}{1}$.
\end{corollary}

\begin{definition}
For any $p\in\N$, $x,y,\lambda\in\R$, define
\[\sigma_p(x,y,\lambda)=\sum_{i=0}^{k-1}\sigma_1(x-i,y+\ln p,\lambda)\]
\end{definition}

\begin{lemma}\label{lem:sigmap}
For any $p\in\N$, $x\in\R$ and $y>0,\lambda>2$,
\[|\intp_p(x)-\sigma_p(x,y,\lambda)|\leqslant1/2+e^{-y}\]
Furthermore if $x<1-\lambda^{-1}$ or $x>p+\lambda^{-1}$ or $d(x,\N)>\lambda^{-1}$ then
\[|\intp_p(x)-\sigma_p(x,y,\lambda)|<e^{-y}\]
and $\sigma_p\in\GSpace{\R^{3}}{p}$.
\end{lemma}

\begin{figure}
\begin{minipage}{0.47\linewidth}%
\begin{center}
\ifthenelse{\importfigures}{%
\includegraphics{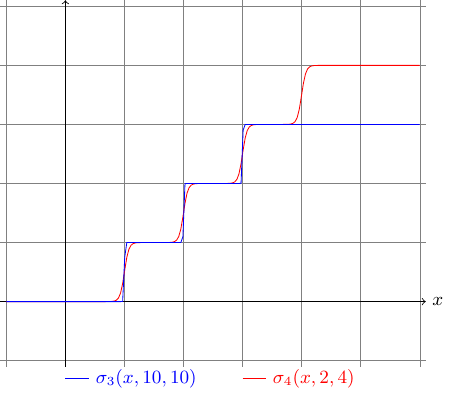}}{%
\begin{tikzpicture}[domain=-1:6,samples=200,scale=1,
    declare function={
        sigma3(\x,\y,\l)=sigma1(\x,\y*3,\l)+sigma1(\x-1,\y*3,\l)+sigma1(\x-2,\y*3,\l);
        sigma4(\x,\y,\l)=sigma1(\x,\y*4,\l)+sigma1(\x-1,\y*4,\l)+sigma1(\x-2,\y*4,\l)+sigma1(\x-3,\y*4,\l);}]
\draw[very thin,color=gray] (-1.1,-1.1) grid (6.1,5.1);
\draw[->] (0,-1.1) -- (0,5.1);
\draw[->] (-1.1,0) -- (6.1,0) node[right] {$x$};
\draw[color=red] plot[id=fn_sigma4_1] function{\fnsigmafour{x}{2}{4}};
\draw[color=blue] plot[id=fn_sigma3_1] function{\fnsigmathree{x}{10}{10}};
\draw[color=red] (3,-1.3) -- (3.4,-1.3) node[right] {$\sigma_4(x,2,4)$};
\draw[color=blue] (0,-1.3) -- (0.4,-1.3) node[right] {$\sigma_3(x,10,10)$};
\end{tikzpicture}}
\end{center}
\caption{Graphical representation of $\sigma_p$}
\label{fig:sigma}
\end{minipage}%
\quad
\begin{minipage}{0.45\linewidth}%
\begin{center}
\ifthenelse{\importfigures}{%
\includegraphics{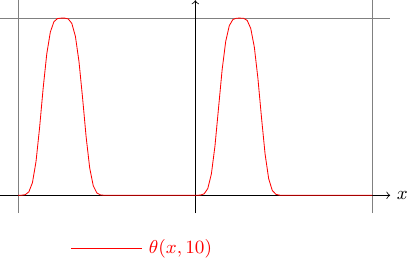}}{%
\begin{tikzpicture}[domain=-1:1,samples=100,scale=3]
\draw[very thin,color=gray] (-1.1,-0.1) grid (1.1,1.1);
\draw[->] (0,-0.1) -- (0,1.1);
\draw[->] (-1.1,0) -- (1.1,0) node[right] {$x$};
\draw[color=red] plot[id=fn_theta_1] function{\fntheta{x}{10}};
\draw[color=red] (-0.7,-0.3) -- (-0.3,-0.3) node[right] {$\theta(x,10)$};
\end{tikzpicture}}
\end{center}
\caption{Graphical representation of $\theta$}
\label{fig:theta}
\end{minipage}%
\end{figure}

Finally, we build a square wave like function which we be useful later on.

\begin{definition}
For any $t\in\R$, and $\lambda>0$, define
$\theta(t,\lambda)=e^{-\lambda\left(1-\sin\left(2\pi t\right)\right)^2}$
\end{definition}

\begin{lemma}\label{lem:theta}
For any $\lambda>0$, $\theta(\cdot,\lambda)$ is a positive and $1$-periodic function bounded by $1$, furthermore
\[\forall t\in[1/2,1], |\theta(t,\lambda)|\leqslant \frac{e^{-\lambda}}{2}\]
\[\int_0^{\frac{1}{2}}\theta(t,\lambda)dt\geqslant\frac{(e\lambda)^{-\frac{1}{4}}}{\pi}\]
and $\theta\in\GSpace{\R\times\R^{*}_{+}}{\lambdafun{t,\lambda}{\max(1,\lambda)}}$.
\end{lemma}

\subsection{Polynomial interpolation}

In order to implement the transition function of the Turing Machine, we will use polynomial interpolation techniques (Lagrange interpolation). But since our simulation may have to deal with some amount of error in inputs, we have to investigate how this error propagates through the interpolating polynomial.

\begin{lemma}\label{lem:diff_prod}
Let $n\in\N$, $x,y\in\R^n$, $K>0$ be such that $\infnorm{x},\infnorm{y}\leqslant K$, then
\[\left|\prod_{i=1}^nx_i-\prod_{i=1}^ny_i\right|\leqslant K^{n-1}\sum_{i=1}^n|x_i-y_i|\]
\end{lemma}

\begin{definition}[Lagrange polynomial]\label{def:lagrange_poly}
Let $d\in\N$ and $f:G\rightarrow\R$ where $G$ is a finite subset of $\R^d$, we define
\[L_f(x)=\sum_{\bar{x}\in G}f(\bar{x})\prod_{i=1}^d\prod_{\substack{y\in G\\y\neq\bar{x}}}\frac{x_i-y_i}{\bar{x}_i-y_i}\]
\end{definition}

We recall that by definition, for all $x\in G$, $L_f(x)=f(x)$ so the interesting part is to know what happen for values of $x$ not in $G$ but close to $G$, that is to relate $L_f(x)-L_f(\tilde{x})$ with $x-\tilde{x}$.

\begin{lemma}\label{lem:interp_L}
Let $d\in\N$, $K>0$ and $f:G\rightarrow\R$, where $G$ is a finite subset of $\R^d$. Then
\[\forall x,z\in[-K,K]^d, |L_f(x)-L_f(z)|\leqslant A\infnorm{x-z}\qquad\text{and}\qquad L_f\in\GSpace{[-K,K]^d}{B}\]
where
\[\delta=\min_{x\neq x'\in G}\min_{i=1}^d|x_i-x'_i|\qquad F=\max_{x\in G}|f(x)|\qquad M=K+\max_{x\in G}\infnorm{x}\]
\[A=|G|F\left(\frac{M}{\delta}\right)^{d(|G|-1)-1}d(|G|-1)\qquad B=|G|F\left(\frac{M}{\delta}\right)^{d(|G|-1)}\]
\end{lemma}

\subsection{Turing Machines --- assumptions}\label{sec:tm_simul_first}

Let $\mathcal{M}=(Q,\Sigma,b,\delta,q_0,F)$ be a Turing Machine which will be fixed for the whole simulation. Without loss of generality we assume that:
\begin{itemize}
\item When the machine reaches a final state, it stays in this state
\item $Q=\intinterv{0}{m-1}$ are the states of the machines; $q_0\in
  Q$ is the initial state; $F\subseteq Q$ are the accepting states
\item $\Sigma=\intinterv{0}{k-2}$ is the alphabet and $b=0$ is the
  blank symbol. 
\item $\delta:Q\times\Sigma\rightarrow Q\times\Sigma\times\{L,R\}$ is the transition function, and we identify $\{L,R\}$ with $\{0,1\}$ ($L=0$ and $R=1$). The components of $\delta$ are denoted by $\delta_1,\delta_2,\delta_3$. That is $\delta(q,\sigma)=(\delta_1(q,\sigma),\delta_2(q,\sigma),\delta_3(q,\sigma))$ where $\delta_1$ is the new state, $\delta_2$ the new symbol and $\delta_3$ the head move direction.
\end{itemize}
Notice that the alphabet of the Turing machine has $k-1$ symbols. This will be important for lemma \ref{lem:config_range}.
Consider a configuration $c=(x,\sigma,y,q)$ of the machine as described in \figref{fig:tm_config}. We can encode it as a triple of integers as done in \cite{GCB08} (e.g. if $x_0,x_1,\ldots$ are the digits of $x$ in base $k$, encode $x$ as the number $x_0+x_1k+x_2k^2+\cdots+x_nk^n$), but this encoding is not suitable for our needs. We define the \emph{rational encoding} $[c]$ of $c$ as follows.

\begin{definition}\label{def:conf_rat_enc}
Let $c=(x,s,y,q)$ be a configuration of $\mathcal{M}$, we define the \emph{rational encoding} $[c]$ of $c$ as $[c]=(0.x,s,0.y,q)$ where:
\[0.x=x_0k^{-1}+x_1k^{-2}+\cdots+x_nk^{-n-1}\in\Q\qquad\text{if}\qquad x=x_0+x_1k+\cdots+x_nk^n\in\N\]
\end{definition}

The following lemma explains the consequences on the rational encoding of configurations of the assumptions we made for $\mathcal{M}$.

\begin{lemma}\label{lem:config_range}
Let $c$ be a reachable configuration of $\mathcal{M}$ and $[c]=(0.x,\sigma,0.y,q)$, then $0.x\in[0,\frac{k-1}{k}]$ and similarly for $0.y$.
\end{lemma}

\begin{figure}
\begin{center}
\ifthenelse{\importfigures}{%
\includegraphics{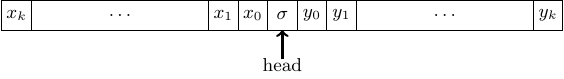}}{%
\begin{tikzpicture}
\draw (0,0) rectangle (0.5,0.5);
\draw (0.25,0.25) node {$\sigma$};

\draw (0.5,0) rectangle (1,0.5);
\draw (0.75,0.25) node {$y_0$};
\draw (1,0) rectangle (1.5,0.5);
\draw (1.25,0.25) node {$y_1$};
\draw (1.5,0) rectangle (4.5,0.5);
\draw (3,0.25) node {$\ldots$};
\draw (4.5,0) rectangle (5,0.5);
\draw (4.75,0.25) node {$y_k$};

\draw (0,0) rectangle (-0.5,0.5);
\draw (-0.25,0.25) node {$x_0$};
\draw (-0.5,0) rectangle (-1,0.5);
\draw (-0.75,0.25) node {$x_1$};
\draw (-1,0) rectangle (-4,0.5);
\draw (-2.5,0.25) node {$\ldots$};
\draw (-4,0) rectangle (-4.5,0.5);
\draw (-4.25,0.25) node {$x_k$};

\draw[->,very thick] (0.25,-0.5) -- (0.25,0);
\draw (0.25,-0.6) node {head};
\end{tikzpicture}}
\end{center}
\caption{Turing Machine configuration}
\label{fig:tm_config}
\end{figure}
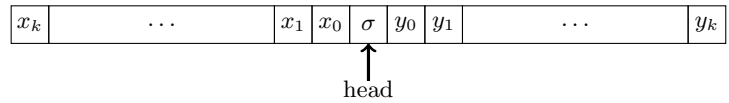

\subsection{Simulation of Turing machines --- step 1: Capturing the transition function}

The first step towards a simulation of a Turing Machine $\mathcal{M}$ using a GPAC is to simulate the transition function of $\mathcal{M}$ with a GPAC-computable function $\operatorname{step}_\mathcal{M}$. The next step is to iterate the function $\operatorname{step}_\mathcal{M}$ with a GPAC. Instead of considering configurations $c$ of the machine, we will consider its rational configurations $[c]$ and use the helper functions defined previously. Theoretically, because $[c]$ is rational, we just need that the simulation works over rationals. But, in practice, because errors are allowed on inputs, the function $\operatorname{step}_\mathcal{M}$ has to simulate the transition function of $\mathcal{M}$ in a manner which tolerates small errors on the input. We recall that $\delta$ is the transition function of the $\mathcal{M}$ and we write $\delta_i$ the $i^{th}$ component of $\delta$.

\begin{definition}
We define:
\[\operatorname{step}_\mathcal{M}:\left\{\begin{array}{ccc}
\R^4&\longrightarrow&\R^4\\
\begin{pmatrix}x\\s\\y\\q\end{pmatrix}&\longmapsto&\begin{pmatrix}
    \operatorname{choose}\left[\fracp(kx),\frac{x+L_{\delta_2}(q,s)}{k}\right]\\
     \operatorname{choose}\left[\intp(kx),\intp(ky)\right]\\
     \operatorname{choose}\left[\frac{y+L_{\delta_2}(q,s)}{k},\fracp(ky)\right]\\
    L_{\delta_1}(q,s)
\end{pmatrix}
\end{array}\right.\]
where
$\operatorname{choose}[a,b]=(1-L_{\delta_3}(q,s))a+L_{\delta_3}(q,s)b$
 and $L_{\delta_i}$ is given by \defref{def:lagrange_poly}.
\end{definition}

The function $\operatorname{step}_\mathcal{M}$ simulates the transition function of the Turing Machine $\mathcal{M}$, as shown in the following result.

\begin{lemma}
Let $c_0,c_1,\ldots$ be the sequence of configurations of $\mathcal{M}$ starting from $c_0$. Then
\[\forall n\in\N, [c_n]=\fiter{\operatorname{step}_\mathcal{M}}{n}([c_0])\]
\end{lemma}

Now we want to extend the function $\operatorname{step}_\mathcal{M}$ to work not only on rationals encodings of configurations but also on reals close to configurations, in a way which tolerates small errors on the input. That is we want to build a robust approximation of $\operatorname{step}_\mathcal{M}$. We already have some results on $L$ thanks to \lemref{lem:interp_L}. We also have some results on $\intp(\cdot)$ and $\fracp(\cdot)$. However, we need to pay attention to the case of nearly empty tapes. This can be done by a shifting $x$ by a small amount ($1/(2k)$) before computing the interger/fractional part. Then lemma \ref{lem:config_range} and lemma \ref{lem:xi} ensure that the result is correct.

\begin{definition}
Define:
\[\overline{\operatorname{step}}_\mathcal{M}(\tau,\lambda):\left\{\begin{array}{ccc}
\R^4&\longrightarrow&\R^4\\
\begin{pmatrix}x\\s\\y\\q\end{pmatrix}&\longmapsto&\begin{pmatrix}
    \operatorname{choose}\left[\overline{\fracp}(kx),\frac{x+L_{\delta_2}(q,s)}{k},q,s\right]\\
     \operatorname{choose}\left[\overline{\intp}(kx),\overline{\intp}(ky),q,s\right]\\
     \operatorname{choose}\left[\frac{y+L_{\delta_2}(q,s)}{k},\overline{\fracp}(ky),q,s\right]\\
    L_{\delta_1}(q,s)
\end{pmatrix}
\end{array}\right.\]
where
\[\operatorname{choose}[a,b,q,s]=(1-L_{\delta_3}(q,s))a+L_{\delta_3}(q,s)b\]
\[\overline{\intp}(x)=\sigma_{k}\left(x+\frac{1}{2k},\tau,\lambda\right)\]
\[\overline{\fracp}(x)=x-\overline{\intp}(x)\]
\end{definition}

We now show that $\overline{\operatorname{step}}_\mathcal{M}$ is a robust version of $\operatorname{step}_\mathcal{M}$. We first begin with a lemma about function $choose$.

\begin{lemma}\label{lem:choose}
There exists $A_3>0$ and $B_3>0$ such that $\forall q,\bar{q},s,\bar{s},a,b,\bar{a},\bar{b}\in\R$, if
\[\infnorm{(\bar{a},\bar{b})}\leqslant M\qquad\text{and}\qquad q\in Q,s\in\Sigma\qquad\text{and}\qquad\infnorm{(q,s)-(\bar{q},\bar{s})}\leqslant 1\]
then
\[\left|\operatorname{choose}[a,b,q,s]-\operatorname{choose}[\bar{a},\bar{b},\bar{q},\bar{s}]\right|\leqslant\infnorm{(a,b)-(\bar{a},\bar{b})}+2MA_3\infnorm{(q,s)-(\bar{q},\bar{s})}\]
Furthermore, $\operatorname{choose}\in\GSpace{\R^2\times[-m,m]\times[-k,k]}{\lambdafun{a,b,q,s}{(1+B_3)(a+b)}}$.
\end{lemma}

\begin{lemma}\label{lem:step_bar}
There exists $A_1,A_2,A_3,B_1,B_2,B_3>0$ such that for any $\tau,\lambda>0$, any valid rational configuration $c=(x,s,y,q)\in\R^4$ and any $\bar{c}=(\bar{x},\bar{s},\bar{y},\bar{q})\in\R^4$, if
\[\infnorm{(x,y)-(\bar{x},\bar{y})}\leqslant \frac{1}{2k^2}-\frac{1}{k\lambda}\qquad\text{and}\qquad\infnorm{(q,s)-(\bar{q},\bar{s})}\leqslant 1\]
then
\[
\begin{array}{r@{\hspace{0.3em}}l}
\text{for }p\in\{1,3\}\quad|\operatorname{step}_\mathcal{M}(c)_p-\overline{\operatorname{step}}_\mathcal{M}(\tau,\lambda)(\bar{c})_p|&\leqslant k\infnorm{(x,y)-(\bar{x},\bar{y})}+(1+2A_3)\left(e^{-\tau}+\frac{A_2}{k}\infnorm{(q,s)-(\bar{q},\bar{s})}\right)\\
|\operatorname{step}_\mathcal{M}(c)_2-\overline{\operatorname{step}}_\mathcal{M}(\tau,\lambda)(\bar{c})_2|&\leqslant2A_3k\infnorm{(q,s)-(\bar{q},\bar{s})}+e^{-\tau}\\
|\operatorname{step}_\mathcal{M}(c)_4-\overline{\operatorname{step}}_\mathcal{M}(\tau,\lambda)(\bar{c})_4|&\leqslant A_1\infnorm{(q,s)-(\bar{q},\bar{s})}
\end{array}\]
Furthermore,
\[\overline{\operatorname{step}}_\mathcal{M}\in\GSpace{(\R_+^*)^2\times[-1,1]\times[-m,m]\times[-1,1]\times[-k,k]}{B_1+(1+B_3)(2k+1+B_2k^{-1})}\]
\end{lemma}

We summarize the previous lemma into the following simpler form.

\begin{corollary}\label{cor:tm_simul_iter}
For any $\tau,\lambda>0$, any valid rational configuration $c=(x,s,y,q)\in\R^4$ and any $\bar{c}=(\bar{x},\bar{s},\bar{y},\bar{q})\in\R^4$, if
\[\infnorm{(x,y)-(\bar{x},\bar{y})}\leqslant \frac{1}{2k^2}-\frac{1}{k\lambda}\qquad\text{and}\qquad\infnorm{(q,s)-(\bar{q},\bar{s})}\leqslant 1\]
then
\[\infnorm{\operatorname{step}_\mathcal{M}(c)-\overline{\operatorname{step}}_\mathcal{M}(\tau,\lambda)(\bar{c})}\leqslant O(1)(e^{-\tau}+\infnorm{c-\bar{c}})\]
Furthermore,
\[\overline{\operatorname{step}}_\mathcal{M}\in\GSpace{(\R_+^*)^2\times[-1,1]\times[-m,m]\times[-1,1]\times[-k,k]}{O(1)}\]
\end{corollary}

\subsection{Simulation of Turing machines --- step 2: Iterating functions with differential equations}\label{sec:it_fn_diff_eq}

We will use a special kind of differential equations to perform the iteration of a map with differential equations. In essence, it relies on the following core differential equation
\begin{equation}\label{eq:branicky_simple}\tag{Reach}
\dot{x}(t)=A\phi(t)(g-x(t))
\end{equation}
We will see that with proper assumptions, the solution converges very quickly to the \emph{goal} g. However, \eqref{eq:branicky_simple} is a simplistic idealization of the system so we need to consider a perturbed equation where the goal is not a constant anymore and the derivative is subject to small errors
\begin{equation}\label{eq:branicky_perturbed}\tag{ReachPerturbed}
\dot{x}(t)=A\phi(t)(\bar{g}(t)-x(t))+E(t)
\end{equation}
We will again see that, with proper assumptions, the solution converges quickly to the \emph{goal} within a small error. Finally we will see how to build a differential equation which iterates a map within a small error.

We first focus on \eqref{eq:branicky_simple} and then \eqref{eq:branicky_perturbed} to show that they behave as expected. In this section we assume $\phi$ is a positive $C^1$ function.

\begin{lemma}\label{lem:branicky_simple}
Let $x$ be a solution of \eqref{eq:branicky_simple}, let $T,\lambda>0$ and assume $A\geqslant\frac{\lambda}{\int_0^T\phi(u)du}$ then $|x(T)-g|\leqslant|g-x(0)|e^{-\lambda}$.
\end{lemma}

\begin{lemma}\label{lem:branicky_perturbed}
Let $T,\lambda>0$ and let $x$ be the solution of \eqref{eq:branicky_perturbed} with initial condition $x(0)=x_0$. Assume $|\bar{g}(t)-g|\leqslant\eta$, $A\geqslant\frac{\lambda}{\int_0^T\phi(u)du}$ and $E(t)=0$ for $t\in[0,T]$. Then
\[|x(T)-g|\leqslant\eta(1+e^{-\lambda})+|x_0-g|e^{-\lambda}\]
\end{lemma}

We can now define a system that simulates the iteration of a function using a system based on \eqref{eq:branicky_perturbed}. It work as described in \cite{GCB08}. There are two variables for simulating each component $f_i$, $i=1,\dots,n$, of the function $f$ to be iterated. There will be periods in which the function is iterated one time. In half of the period, half ($n$) of the variables will stay (nearly) constant and close to values $\alpha_1,\dots,\alpha_n$, while the other remaining $n$ variables update their value to $f_i(\alpha_1,\dots,\alpha_n)$, for $i=1,\dots,n$. In the other half period, the second subset of variables is then kept constant, and now it is the first subset of variables which is updated to $f_i(\alpha_1,\dots,\alpha_n)$, for $i=1,\dots,n$.

\begin{definition}
Let $d\in\N$, $F:\R^d\rightarrow\R^d$, $\lambda\geqslant1,\mu\geqslant0$ and $u_0\in\R^d$, we define
\begin{equation}\label{eq:simul_iter}\tag{Iterate}
\left\{\begin{array}{r@{}l}
z(0)&=u_0\\
u(0)&=u_0\\
\end{array}\right.
\qquad
\left\{\begin{array}{r@{}l}
\dot{z}_i(t)&=A\theta(t,B)(F_i(u(t))-z_i(t))\\
\dot{u}_i(t)&=A\theta(t-1/2,B)(z_i(t)-u_i(t))
\end{array}\right.
\end{equation}
where $A=10(\lambda+\mu)^2$ and $B=4(\lambda+\mu)$.
\end{definition}

\begin{theorem}\label{th:fn_simul_diff_eq}
Let $d\in\N$, $F:\R^d\rightarrow\R^d$, $\lambda\geqslant1$, $\mu\geqslant0$, $u_0,c_0\in\R^d$. Assume $z,u$ are solutions to \eqref{eq:simul_iter} and let $\Delta F$ and $M\geqslant 1$ be such that
\[\forall k\in\N,\forall \varepsilon>0,\forall x\in]-\varepsilon,\varepsilon[^d,\infnorm{\fiter{F}{k+1}(c_0)-F\left(\fiter{F}{k}(c_0)+x\right)}\leqslant\Delta F(\varepsilon)\]
\[\forall t\geqslant 0, \infnorm{u(t)},\infnorm{z(t)},\infnorm{F(u(t))}\leqslant M=e^\mu\]
and consider
\[\left\{\begin{array}{@{}c@{}l}
    \varepsilon_0&=\infnorm{u_0-c_0}\\
    \varepsilon_{k+1}&=(1+3e^{-\lambda})\Delta F(\varepsilon_k+2e^{-\lambda})+5e^{-\lambda}
    \end{array}\right.\]
Then
\[\forall k\in\N,\infnorm{u(k)-\fiter{F}{k}(c_0)}\leqslant\varepsilon_k\]
Furthermore, if $F\in\GSpace{[-M,M]^d}{s_F}$ for $s_F:[-M,M]\rightarrow\R$ then \[\left(\lambdafun{\lambda,\mu,t,u_0}{u(t)}\right)\in\GSpace{\left(\R_+^*\right)^{3}\times\R^d}{\lambdafun{\lambda,\mu,t,u_0}{\max(1,4(\lambda+\mu),s_F(M))}}\]
\end{theorem}

\subsection{Simulation of Turing machines --- step 3: Putting all pieces together}

In this section, we will use results of both \secref{sec:tm_simul_first} and \secref{sec:it_fn_diff_eq} to simulate Turing Machines with differential equations. Indeed, in \secref{sec:tm_simul_first} we showed that we could simulate a Turing Machine by iterating a robust real map, and in \secref{sec:it_fn_diff_eq} we showed how to efficiently iterate a robust map with differential equations. Now we just have to put these results together.

\begin{lemma}\label{lem:rec_seq_geom_arith}
Let $a>1$ and $b\geqslant0$, assume $u\in\R^\N$ satisfies $u_{n+1}\leqslant au_n+b$, $n\geqslant0$. Then
\[u_n\leqslant a^nu_0+b\frac{a^n-1}{a-1},\quad n\geqslant0\]
\end{lemma}

\begin{theorem}\label{th:tm_simual_diff_eq}
Let $\mathcal{M}$ be a Turing Machine as in \secref{sec:tm_simul_first}, then there are functions $s_f:I\to\R^{4}$ and  $f_\mathcal{M}\in\GSpace{\R^4}{s_f}$ such that for any sequence $c_0,c_1,\ldots,$ of configurations of $\mathcal{M}$ starting with input $e$:
\[\forall S, T\in\R_+^*,\forall n\leqslant T,\infnorm{[c_n]-f_\mathcal{M}(S,T,n,e)}\leqslant e^{-S}\]
and
\[\forall S, T\in\R_+^*,\forall n\leqslant T, s_f(S,T,n,e)=O(poly(S,T))\]
\end{theorem}

\begin{proof}
Let $\lambda>0$ (to be fixed later) and apply \thref{th:fn_simul_diff_eq} to $F=\overline{\operatorname{step}}_\mathcal{M}(\lambda,4k)$. By \corref{cor:tm_simul_iter}, $\exists K_1,K_2$ such that
\[\Delta F(\varepsilon)=K_1(e^{-\tau}+\varepsilon)\]
and
\[\forall x\in\Lambda=[-1,1]\times[-m,m]\times[-1,1]\times[-k,k], \infnorm{F(x)}\leqslant K_2\]
Let $M=K_2+1$. The recurrence relation of $\varepsilon$ (where $u,z$ are defined as in \eqref{eq:simul_iter})
\[\left\{\begin{array}{@{}c@{}l}
    \varepsilon_0&=\infnorm{u(0)-c_0}\\
    \varepsilon_{k+1}&=(1+3e^{-\lambda})\Delta F(\varepsilon_k+2e^{-\lambda})+5e^{-\lambda}
    \end{array}\right.\]
now simplifies to (using that $e^{-\lambda}\leqslant1$)
\begin{align*}
\varepsilon_{k+1}
  &\leqslant(1+3e^{-\lambda})K_1(e^{-\tau}+\varepsilon_k+2e^{-\lambda})+5e^{-\lambda}\\
  &\leqslant K_1(1+3e^{-\lambda})\varepsilon_k+2K_1(1+3e^{-\lambda})e^{-\lambda}+5e^{-\lambda}\\
  &\leqslant \underbrace{K_1(1+3e^{-\lambda})}_{a}\varepsilon_k+\underbrace{(8K_1+5)e^{-\lambda}}_{b}
\end{align*}
Now apply \lemref{lem:rec_seq_geom_arith} to get an explicit expression
\[\varepsilon_n\leqslant a^nu_0+b\frac{a^n-1}{a-1}\]
If we take as initial condition the exact rational configuration $[c_0]$, we immediately get that $u_0=0$. Let $K_3=4K_1$, then $a\leqslant K_3$. Pick $\lambda=S+T\log(K_3)+\log(8K_1+5)$. Then $\varepsilon_T\leqslant e^{-S}$.

We check with \thref{th:fn_simul_diff_eq} that $\infnorm{u(t)},\infnorm{z(t)}\leqslant M$ for $t\leqslant T$ since $\varepsilon_T\leqslant1$.

Finally, $f_\mathcal{M}=\left(\lambdafun{\lambda,\mu,t,e}{u(t)}\right)\in\GSpace{\R^4}{\underbrace{\lambdafun{\lambda,\mu,t,e}{\max(1,4(\lambda+\mu),s_F(M))}}_{s_f}}$ and $s_f=O(poly(S,T))$.
\end{proof}

\section{Acknowledgments} D.S. Gra\c{c}a was partially supported by
\emph{Funda\c{c}\~{a}o para a Ci\^{e}ncia e a Tecnologia} and EU FEDER
POCTI/POCI via SQIG - Instituto de Telecomunica\c{c}\~{o}es through
the FCT project PEst-OE/EEI/LA0008/2011.

Olivier Bournez and Amaury Pouly were partially supported by ANR
project SHAMAN, by DGA, and by DIM LSC DISCOVER project.

\bibliographystyle{plain}
\bibliography{ContComp}
\end{document}